\documentclass[american,aps,pra,reprint,superscriptaddress]{revtex4-1}
\usepackage[T1]{fontenc}
\usepackage[latin9]{inputenc}
\usepackage{color}
\usepackage{babel}
\usepackage{amsmath}
\usepackage{amsthm}
\usepackage{graphicx}
\usepackage{amssymb}

\usepackage[unicode=true,pdfusetitle,
 bookmarks=true,bookmarksnumbered=false,bookmarksopen=false,
 breaklinks=false,pdfborder={0 0 0},backref=false,colorlinks=true,allcolors=magenta]
 {hyperref}

\makeatletter
\theoremstyle{plain}
\newtheorem{thm}{\protect\theoremname}
  \theoremstyle{plain}
  \newtheorem{lem}[thm]{\protect\lemmaname}

\usepackage{babel}
\usepackage{txfonts}
\usepackage{braket}
\usepackage{enumerate}
\usepackage{colortbl}

\providecommand{\lemmaname}{Lemma}
\providecommand{\theoremname}{Theorem}

\makeatother

  \providecommand{\lemmaname}{Lemma}
\providecommand{\theoremname}{Theorem}

\begin{document}

\title{Rates of multi-partite entanglement transformations and 
applications in quantum networks}

\author{A. Streltsov}

\affiliation{Centre for Quantum Optical Technologies IRAU, Centre of New Technologies, University of Warsaw, Banacha 2c, 02-097 Warsaw, Poland}

\affiliation{Faculty of Applied Physics and Mathematics, Gda\'{n}sk University
of Technology, 80-233 Gda\'{n}sk, Poland}

\affiliation{National Quantum Information Centre in Gda\'{n}sk, 81-824 Sopot,
Poland}

\address{Dahlem Center for Complex Quantum Systems, Freie Universit\"at Berlin, 14195 Berlin, Germany}

\author{C. Meignant}

\address{Laboratoire d'Informatique de Paris 6, CNRS, Sorbonne Universit\'e, 4 place Jussieu, 75005 Paris, France}

\address{Dahlem Center for Complex Quantum Systems, Freie Universit\"at Berlin, 14195 Berlin, Germany}

\author{J. Eisert}

\address{Dahlem Center for Complex Quantum Systems, Freie Universit\"at Berlin, 14195 Berlin, Germany}
\begin{abstract}
The theory of the asymptotic manipulation of pure bipartite quantum
systems can be considered completely understood: The rates at which
bipartite entangled states can be asymptotically transformed into
each other are fully determined by a single number each, the respective
entanglement entropy. In the multi-partite setting, similar questions
of the optimally achievable rates of transforming one pure state into
another are notoriously open. This seems particularly unfortunate
in the light of the revived interest in such questions due to the
perspective of experimentally realizing multi-partite quantum networks.
In this work, we report substantial progress by deriving surprisingly
simple upper and lower bounds on the rates that can be achieved in
asymptotic multi-partite entanglement transformations. These bounds
are based on ideas of entanglement combing and state merging. We identify
cases where the bounds coincide and hence provide the exact rates.
As an example, we bound rates at which resource states for the cryptographic
scheme of quantum secret sharing can be distilled from arbitrary pure
tripartite quantum states, providing further scope for quantum internet applications
beyond point-to-point. 
\end{abstract}
\maketitle
Entanglement is the feature of quantum mechanics that renders it distinctly
different from a classical theory~\cite{Horodecki2009}. It is at the heart of quantum information
science and technology as a resource that is used to accomplish task
(and is increasingly also seen as an important concept in condensed-matter
physics). Given its significance in protocols of quantum information,
it hardly surprises that already early in the development of the field,
questions were asked how one form of entanglement could be transformed
into another. It was one of the early main results of the field of
quantum information theory to show that all pure bipartite states
could be asymptotically reversibly transformed to maximally entangled
states with local operations and classical communication (LOCC) at
a rate that is determined by a single number \cite{PureBipartiteBennett}:
the entanglement entropy, the von-Neumann entropy of each reduced
state. This insight makes the resource character of bipartite entanglement
most manifest: The entanglement content is given simply by its content
of maximally entangled states, and each form can be transformed reversibly
into another and back.

The situation in the multi-partite setting is significantly more intricate,
however \cite{MultipartiteReview,PhysRevA.63.012307,PhysRevA.95.012323}. The rates that can be achieved
when aiming at asymptotically transforming one multi-partite state
into another with LOCC are far from clear. It is not even understood
what the ``ingredients'' of multi-partite entanglement theory are 
\cite{PhysRevA.63.012307,PhysRevA.72.059907},
so the basic units of multi-partite entanglement from which any other pure state
can be asymptotically reversibly prepared. This state of affairs is
unfortunate, and even more so since multi-partite states come again
more into the focus of attention in the light of the observation that
elements of the vision of a quantum network -- or the ``quantum internet'' \cite{QuantumInternet}
-- may become an experimental reality in the not too far future. It
is not that multi-partite entanglement ceases to have a resource character:
For example, Greenberger-Horne-Zeilinger (GHZ) states are known to
constitute a resource for quantum secret sharing \cite{PhysRevA.59.1829,SecretSharing},
the probably best known multi-partite cryptographic primitive. Progress
on stochastic conversion for several copies of multi-partite states
was made recently \cite{WAndGHZ,FromGHZ}. However, given a collection
of arbitrary pure states, it is not known at what rate such states
could be asymptotically distilled under LOCC.


In this work, we report surprisingly substantial progress on the old
question of the rate at which GHZ and other multi-partite states can
be asymptotically distilled from arbitrary pure states. Surprising,
in that much of the technical substance can be delegated to the powerful
machinery of entanglement combing~\cite{Yang2009}, putting it here into a fresh context,
which in turn can be seen to derive from quantum
state merging \cite{Horodecki2005,Horodecki2007}, assisted entanglement
distillation~\cite{DiVincenzo1999,Smolin2005}, and time-sharing,
meaning, using resource states in different roles in the asymptotic
protocol. The basic insight underlying the analysis is that entanglement
combing provides a reference, a helpful normal form rooted in the
better understood theory of bipartite entanglement, that can be used
in order to assess rates of asymptotic multi-partite state conversion.
Basically, putting entanglement combing to good work, therefore, we
are in the position to make significant progress on the question of
entanglement transformation rates in a general setting.

\medskip
\textbf{\emph{Multi-partite state conversion. }}We consider the problem
of converting an $n$-partite state $\rho$ into $\sigma$ via $n$-partite LOCC. In particular, we are interested in the optimally achievable
asymptotic rate for this procedure, which can be formally defined
as 
\begin{equation}
R(\rho\rightarrow\sigma)=\sup\left\{ r:\lim_{k\rightarrow\infty}\left(\inf_{\Lambda}\left\Vert \Lambda\left(\rho^{\otimes k}\right)-\sigma^{\otimes\left\lfloor rk\right\rfloor }\right\Vert _{1}\right)=0\right\} .\label{eq:R}
\end{equation}
Here, $\Lambda$ reflects an $n$-partite LOCC operation and $||M||_{1}=\mathrm{Tr}\sqrt{M^{\dagger}M}$
denotes the trace norm. This problem has a known solution in the bipartite
case $n=2$ for conversion between arbitrary pure states $\psi^{AB}\rightarrow\phi^{AB}$,
rooted in Shannon theory. The corresponding rate in this case can
be written as~\cite{PureBipartiteBennett} 
\begin{equation}
R(\psi^{AB}\rightarrow\phi^{AB})=\frac{S(\psi^{A})}{S(\phi^{A})},\label{eq:bipartite}
\end{equation}
where $S(\rho)=-\mathrm{Tr}(\rho\log_{2}\rho)$ is the von Neumann
entropy. Moreover, $\psi^{AB}$ indicates that the state is shared between parties
referred to as Alice and Bob, while $\psi^{A}$ reflects the reduced
state of Alice.

This simple picture ceases to hold in any setting beyond the bipartite
one. Indeed, significantly less is known in the multi-partite setting
for $n\geq3$ \cite{MultipartiteReview}.
Needless
to say, the bipartite solution~(\ref{eq:bipartite}) readily gives
upper bounds on the rates in multi-partite settings. For example, for
conversion between tripartite pure states $\psi^{ABC}\rightarrow\phi^{ABC}$,
it must be true that 
\begin{equation}
R(\psi^{ABC}\rightarrow\phi^{ABC})\leq\min\left\{ \frac{S(\psi^{A})}{S(\phi^{A})},\frac{S(\psi^{B})}{S(\phi^{C})},\frac{S(\psi^{C})}{S(\phi^{C})}\right\} .\label{eq:upper-bound}
\end{equation}
This follows from the fact that any tripartite LOCC protocol is also
bipartite with respect to any of the bipartitions. If the desired
final state $\phi^{ABC}$ is the GHZ state with state vector $\ket{\mathrm{GHZ}}=(\ket{000}+\ket{111})/\sqrt{2}$,
the bound in Eq.~(\ref{eq:upper-bound}) is known to be achievable
whenever one of the reduced states $\psi^{AB}$, $\psi^{BC}$ or $\psi^{AC}$
is separable~\cite{Smolin2005}. 

We also note that for some states the bound in Eq.~(\ref{eq:upper-bound}) is a strict inequality. This can be seen by considering
the scenario where each of the parties holds two qubits respectively. Consider now the transformation
\begin{equation}
\begin{split}\ket{\mathrm{GHZ}}^{A_{1}B_{1}C_{1}}\otimes & \ket{\mathrm{GHZ}}^{A_{2}B_{2}C_{2}}\rightarrow\\
 & \ket{\Phi^{+}}^{A_{1}B_{1}}\otimes\ket{\Phi^{+}}^{A_{2}C_{1}}\otimes\ket{\Phi^{+}}^{B_{2}C_{2}},
\end{split}
\end{equation}
i.e., the parties aim to transform two GHZ states into Bell states $\ket{\Phi^+}=(\ket{00}+\ket{11})/\sqrt{2}$ which
are equally distributed among all the parties. It is straightforward
to check that in this case the bound in Eq.~(\ref{eq:upper-bound})
becomes $R\leq1$. However, the bound is not achievable, as the aforementioned
transformation cannot be performed with unit rate~\cite{Linden2005}.

\medskip
\textbf{\emph{Lower bound on conversion rates for three parties}.} The above discussion
suggests that the bound in Eq.~(\ref{eq:upper-bound}) is a very
rough estimate for general transformations and is saturated only for
very specific sets of states, having zero volume in the set of all
pure states. Quite surprisingly, we will see below that this is not
the case: there exist large families of tripartite pure states which
saturate the bound~(\ref{eq:upper-bound}). This will follow from
a very general and surprisingly simple lower bound on conversion rate,
which will be presented below in Theorem~\ref{thm:bound}.

The methods developed here build upon the machinery of \emph{entanglement
combing}, which was introduced and studied for general $n$-partite
scenarios in Ref.\ \cite{Yang2009}. In the specific tripartite setting, entanglement combing aims to transform the initial state $\psi^{ABC}$
into a state of the form $\mu^{A_{1}B}\otimes\nu^{A_{2}C}$ with pure
bipartite states $\mu$ and $\nu$. The following Lemma restates the
results from Ref.\ \cite{Yang2009} in a form which will be suitable
for the purpose of this work.
\begin{lem}[Conditions from tripartite entanglement combing]
\label{lem:combing}
The transformation 
\begin{equation}
\psi^{ABC}\rightarrow\mu^{A_{1}B}\otimes\nu^{A_{2}C}\label{eq:combing-1}
\end{equation}
is possible via asymptotic LOCC if and only if \begin{subequations}\label{eq:combing-2}
\begin{align}
E(\mu^{A_{1}B})+E(\nu^{A_{2}C}) & \leq S(\psi^{A}),\\
E(\mu^{A_{1}B}) & \leq S(\psi^{B}),\\
E(\nu^{A_{2}C}) & \leq S(\psi^{C}).
\end{align}
\end{subequations} 
\end{lem}
\noindent We refer to Appendix~\ref{sec:ProofLemma} for the proof of the Lemma. Using
this result, we are now in position to present a tight lower bound
on the transformation rate between tripartite pure states. 

\begin{figure}
\begin{center}
	\includegraphics[width=1\linewidth]{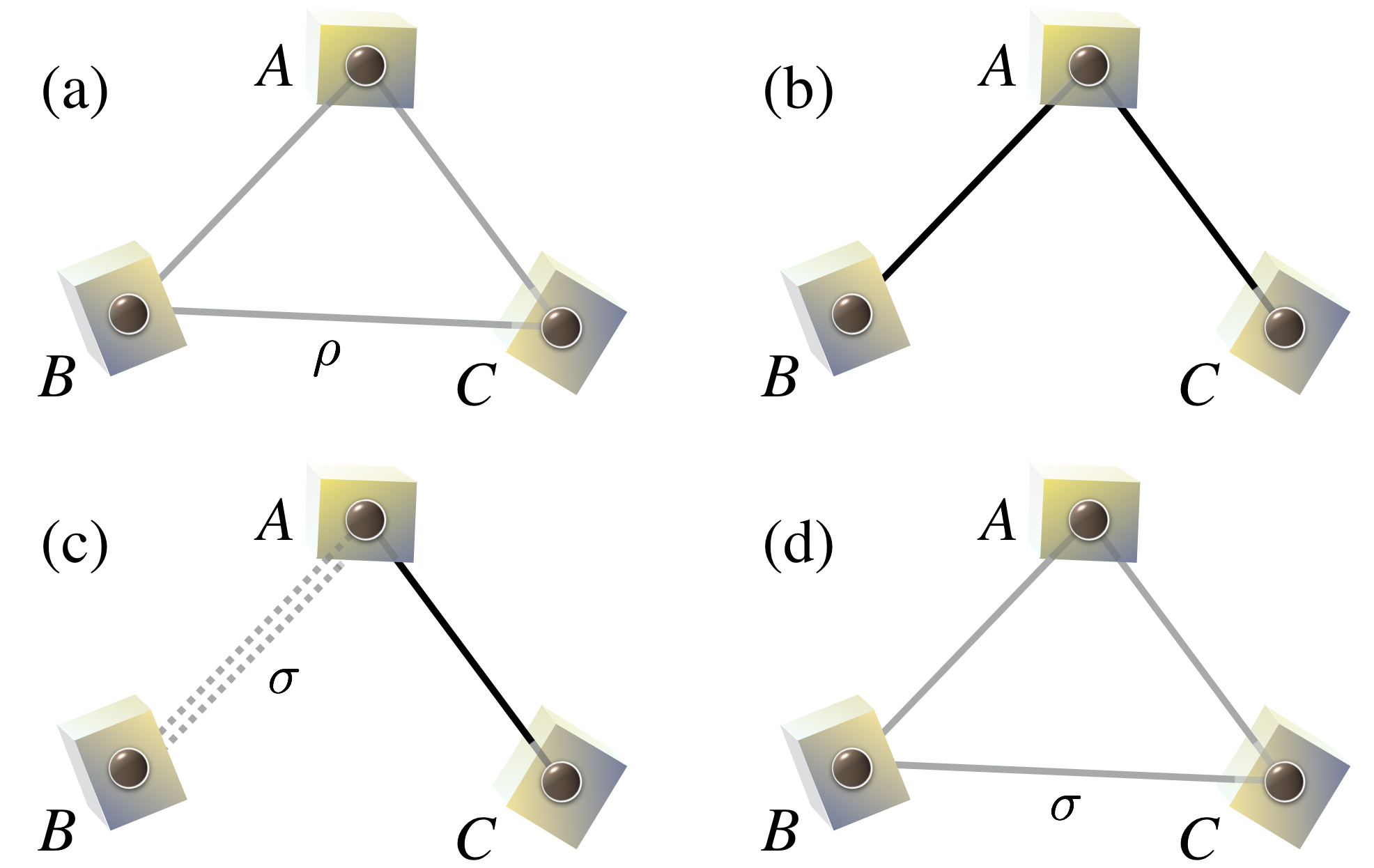}
\end{center}
\caption{\label{fig:Conversion}Conversion of a multi-partite resource state $\rho$ (a) into the desired final state $\sigma$ (d). The conversion is achieved via entanglement combing, i.e., via transforming the initial state $\rho$ into singlets [black solid lines in (b)]. One of the singlets is then converted into the desired final state $\sigma$ [gray dotted lines in (c)]. The remaining singlets [black solid line in (c)] are then used for teleporting the parts of $\sigma$ to the remaining parties.}
\end{figure}

\begin{thm}[Lower bound for state transformations]
\label{thm:bound}For tripartite pure states $\psi^{ABC}$ and $\phi^{ABC}$,
the LOCC conversion rate is bounded from below as 
\begin{equation}
R(\psi^{ABC}\rightarrow\phi^{ABC})\geq\min\left\{ \frac{S(\psi^{A})}{S(\phi^{B})+S(\phi^{C})},\frac{S(\psi^{B})}{S(\phi^{B})},\frac{S(\psi^{C})}{S(\phi^{C})}\right\} .\label{eq:bound-1}
\end{equation}
\end{thm}
\begin{proof}
We prove this bound by presenting an explicit protocol achieving the
bound, which is also summarized in Fig.~\ref{fig:Conversion}. In the first step, the parties apply entanglement combing $\psi^{ABC}\rightarrow\mu^{A_{1}B}\otimes\nu^{A_{2}C}$
in such a way that the following equalities are fulfilled for some
$r\geq0$, 
\begin{equation}
E(\mu^{A_{1}B})=rS(\phi^{B}),\,\,\,\,\,\,\,\,E(\nu^{A_{2}C})=rS(\phi^{C}).\label{eq:rates}
\end{equation}
The significance of this specific choice will become clear in a moment.
In the next step, Alice and Charlie apply LOCC for transforming the
state $\nu^{A_{2}C}$ into the desired final state $\phi^{A_{2}A_{3}C}$.
Since this is a bipartite LOCC protocol, the rate for this process
is given by $E(\nu^{A_{2}C})/S(\phi^{C})$. Note that due to Eqs.~(\ref{eq:rates}),
this rate is equal to $r$.

In a next step, Alice applies what is called Schumacher compression~\cite{Schumacher1995}
to her register $A_{3}$. The overall compression rate per copy of
the initial state $\psi^{ABC}$ is given as 
\begin{equation}
\tilde{r}=rS(\phi^{A_{3}})=rS(\phi^{B}),
\end{equation}
where in the last equality we used the fact that $S(\phi^{A_{3}})=S(\phi^{B})$.
Due to Eqs.~(\ref{eq:rates}), this rate interestingly coincides
with the entanglement of the state $\mu^{A_{1}B}$, 
\begin{equation}
\tilde{r}=E(\mu^{A_{1}B}).\label{eq:compression}
\end{equation}
In a final step, Alice and Bob distill the states $\mu^{A_{1}B}$
into maximally entangled bipartite singlets, and use them to teleport~\cite{BennettTeleportation,TeleportationReview} the (compressed)
particle $A_{3}$ to Bob. Due to Eq.~(\ref{eq:compression}), Alice
and Bob share exactly the right amount of entanglement for this procedure,
i.e., the process is possible with rate one and no entanglement is
left over. In summary, the overall protocol transforms the state $\psi^{ABC}$
into $\phi^{ABC}$ at rate $r$.

For completing the proof, we will now show that $r$ can be chosen
such that 
\begin{equation}
r=\min\left\{ \frac{S(\psi^{A})}{S(\phi^{B})+S(\phi^{C})},\frac{S(\psi^{B})}{S(\phi^{B})},\frac{S(\psi^{C})}{S(\phi^{C})}\right\} .\label{eq:r}
\end{equation}
This can be seen directly by inserting Eqs.~(\ref{eq:rates}) into
Eqs.~(\ref{eq:combing-2}). In particular, the rate $r$ can attain
any value which is simultaneously compatible with inequalities 
\begin{equation}
r\leq\frac{S(\psi^{A})}{S(\phi^{B})+S(\phi^{C})},\,\,\,r\leq\frac{S(\psi^{B})}{S(\phi^{B})},\,\,\,r\leq\frac{S(\psi^{C})}{S(\phi^{C})}.
\end{equation}
This completes the proof of the theorem. 
\end{proof}
We stress some important aspects and implications of this theorem.
Whenever the minimum
in Eq.~(\ref{eq:bound-1}) is attained on the second or third entry,
the lower bound coincides with the upper bound in Eq.~(\ref{eq:upper-bound}).
This means that in all these instances the conversion problem is completely
solved, giving rise to the rate 
\begin{equation}
R(\psi^{ABC}\rightarrow\phi^{ABC})=\min\left\{ \frac{S(\psi^{A})}{S(\phi^{A})},\frac{S(\psi^{B})}{S(\phi^{B})},\frac{S(\psi^{C})}{S(\phi^{C})}\right\} . \label{eq:tight}
\end{equation}
Moreover, the bound in Eq.~(\ref{eq:bound-1}) can be immediately
generalized by interchanging the roles of the parties, i.e., 
\begin{align}
\begin{split}R(\psi^{ABC} \rightarrow\phi^{ABC})\geq & \min\left\{ \frac{S(\psi^{B})}{S(\phi^{A})+S(\phi^{C})},\frac{S(\psi^{A})}{S(\phi^{A})},\frac{S(\psi^{C})}{S(\phi^{C})}\right\} ,\label{eq:bound-2}
\end{split}
\\
\begin{split}R(\psi^{ABC} \rightarrow\phi^{ABC})\geq & \min\left\{ \frac{S(\psi^{C})}{S(\phi^{A})+S(\phi^{B})},\frac{S(\psi^{A})}{S(\phi^{A})},\frac{S(\psi^{B})}{S(\phi^{B})}\right\} .\label{eq:bound-3}
\end{split}
\end{align}
The best bound is obtained by taking the maximum of Eqs.~(\ref{eq:bound-1}),
(\ref{eq:bound-2}) and (\ref{eq:bound-3}).

Our results also shed new light on reversibility questions for tri-partite
state transformations. In general, a transformation $\psi\rightarrow\phi$
is said to be reversible if the conversion rates fulfill the relation
\begin{equation}
R(\psi\rightarrow\phi)={R(\phi\rightarrow\psi)}^{-1}.
\end{equation}
Let now $\psi$ and $\phi$ be two states for which the bound in Theorem~\ref{thm:bound}
is tight, e.g., $R(\psi\rightarrow\phi)=S(\psi^{A})/S(\phi^{A})$.
Due to Eq.~(\ref{eq:upper-bound}) it must be that $S(\psi^{A})/S(\phi^{A})\leq S(\psi^{B})/S(\phi^{B})$
in this case. If this inequality is strict (which will be the generic
case), we obtain for the inverse transformation $\phi\rightarrow\psi$
\begin{equation}
R(\phi\rightarrow\psi)\leq\frac{S(\phi^{B})}{S(\psi^{B})}<\frac{S(\phi^{A})}{S(\psi^{A})}={R(\psi\rightarrow\phi)}^{-1},
\end{equation}
where the first inequality follows from Eq.~(\ref{eq:upper-bound}).
These results show that those states which saturate the bound~(\ref{eq:upper-bound})
do not allow for reversible transformations in the generic case.

We will now comment on the limits of the approach presented here.
In particular, it is important to note that the lower bound in Theorem~\ref{thm:bound}
is not optimal in general. This can be seen in the most simple way
by considering the trivial transformation which leaves the state unchanged,
i.e., $\psi^{ABC}\rightarrow\psi^{ABC}$. 
Clearly, this can be achieved
with unit rate $R=1$. However, if we apply the lower bound in Theorem~\ref{thm:bound}
to this transformation, we get $R\geq{S(\psi^{A})}/[S(\psi^{B})+S(\psi^{C})]$.
Due to subadditivity, it follows that that our lower bound is in general below the achievable unit rate in this case.

\medskip
\textbf{\emph{Multi-partite pure 
states.}} In the discussion
so far, we have focused on tripartite pure states. However, the presented
tools can readily be applied to more general scenarios involving an
arbitrary number of parties. In this more general setup the parties will be called Alice $(A)$ and $N$ Bobs $(B_i)$ with $1\leq i \leq N$. The aim of the process in this case is the asymptotic conversion of the $N+1$-partite pure state $\psi = \psi^{AB_1 \ldots B_N}$ into the state $\phi = \phi^{AB_1 \ldots B_N}$. The general idea for this procedure follows the same line of reasoning as in the tripartite scenario discussed above. In the first step, entanglement combing is applied to the
state $\psi$, i.e., the transformation 
\begin{equation}
\psi \rightarrow\mu_{1}^{A_{1}B_1}\otimes\mu_{2}^{A_{2}B_2}\otimes\cdots\otimes\mu_N^{A_NB_N}
\end{equation}
with pure states $\mu_{i}$. In the next step, Alice and the first Bob $B_1$ transform
their state $\mu_{1}^{A_{1}B_1}$ into the desired final state $\phi$
via bipartite LOCC. In the final step, Alice applies Schumacher compression
to parts of her state $\phi$, and sends these parts to each of the remaining Bobs
$B_2,\ldots,B_N$ by using entanglement obtained in the first step of
this protocol. As in the tripartite case, this protocol can be further optimized by interchanging the roles of the parties and applying the time-sharing technique. 

\begin{thm}[Lower bound for multi-partite state conversion]
\label{thm:bound2}For $N+1$-partite pure states $\psi^{AB_{1}\dots B_{N}}$
and $\phi^{AB_{1}\dots B_{N}}$, the LOCC conversion rate is bounded from below as
\begin{equation}
R(\psi^{AB_{1}\dots B_{N}}\rightarrow\phi^{AB_{1}\dots B_{N}})\geq\min_X\left\{ \frac{S(\psi^{AX})}{\sum_{B_{i}\notin X}S(\phi^{B_{i}})}\right\},\label{eq:lower-bound}
\end{equation}
where $X$ denotes a subsystem of all Bobs.
\end{thm}

The theorem is proven in Appendix \ref{ap:nconver}. By using similar arguments as below Eq.~(\ref{eq:upper-bound}), an upper bound to the conversion rate is found to be
\begin{equation}
R(\psi^{AB_{1}\dots B_{N}}\rightarrow\phi^{AB_{1}\dots B_{N}})\leq\min_i\frac{S(\psi^{B_{i}})}{S(\phi^{B_{i}})}.\label{eq:upper-bound-mult}
\end{equation}
 The bounds in Eqs.~(\ref{eq:lower-bound}) and (\ref{eq:upper-bound-mult}) coincide if the following equality holds true for some 
$1\leq i \leq N$,
\begin{equation}
\frac{S(\psi^{B_i})}{S(\phi^{B_i})}=\min_X\left\{ \frac{S(\psi^{AX})}{\sum_{B_{j}\notin X}S(\phi^{B_{j}})}\right\}.
\end{equation}
In those instances, Theorem \ref{thm:bound2} leads to a full solution of the conversion problem, and the corresponding rate is given by
\begin{equation}
R(\psi^{AB_{1}\dots B_{N}} \rightarrow \phi^{AB_{1}\dots B_{N}})=\min_i\frac{S(\psi^{B_{i}})}{S(\phi^{B_{i}})}.
\end{equation}
Again, as in the tripartite case, the bound of Eq.~(\ref{eq:lower-bound})
can be generalized by interchanging the roles of Alice and different Bobs.

\medskip
\textbf{\emph{Generalization to multi-partite mixed states.}} 
We will now show that the ideas which led to lower bounds on conversion rates in the previous sections can also be used in this mixed-state scenario. We will demonstrate this on a specific example, considering the transformation
\begin{equation}
\ket{\mathrm{GHZ}}\!\bra{\mathrm{GHZ}}\rightarrow\sigma,
\end{equation}
where $\ket{\mathrm{GHZ}}=(\ket{0}^{\otimes N+1}+\ket{1}^{\otimes N+1})/\sqrt{2}$
denotes an $N+1$-partite GHZ state vector, and $\sigma=\sigma^{AB_1\ldots B_N}$
is an arbitrary $N+1$-partite mixed state. As we show in Appendix~\ref{sec:GHZ},
by using similar methods as in previous sections, we obtain a lower bound on the transformation
rate, 
\begin{equation}
R(\ket{\mathrm{GHZ}}\!\bra{\mathrm{GHZ}}\ \rightarrow\sigma)\geq\frac{1}{E_{\mathrm{c}}^{A|B_1\ldots B_N}(\sigma)+\sum_{j=3}^N S(\sigma^{B_j})},
\end{equation}
where $E_{\mathrm{c}}^{A|B_1\ldots B_N}$ denotes the entanglement cost \cite{EntCost}
between Alice and all the other Bobs.

The upper bound (\ref{eq:upper-bound-mult}) for the transformation rate $R$ can be generalized as (see Eq.~(146) in Ref.~\cite{Horodecki2009})
\begin{equation}
R(\rho\rightarrow\sigma)\leq\min_{\mathcal{P}}\frac{E_{\infty}^{\mathcal{P}|\overline{\mathcal{P}}}(\rho)}{E_{\infty}^{\mathcal{P}|\overline{\mathcal{P}}}(\sigma)}.\label{eq:upper-bound-2}
\end{equation}
Here, $E_{\infty}(\rho)=\lim_{n\rightarrow\infty}E_{\mathrm{r}}(\rho^{\otimes n})/n$
is the regularized relative entropy of entanglement \cite{AsymptoticRelent,Winter2016},
and $\mathcal{P}|\overline{\mathcal{P}}$ denotes a bipartition of
all the $N+1$ subsystems~\footnote{If there is a bipartition $\mathcal{P}|\overline{\mathcal{P}}$ with
$E_{\infty}^{\mathcal{P}|\overline{\mathcal{P}}}(\rho)=E_{\infty}^{\mathcal{P}|\overline{\mathcal{P}}}(\sigma)=0$,
this bipartition is not taken into account in Eq.~(\ref{eq:upper-bound-2}).}.

\medskip
\textbf{\emph{Applications in quantum networks.}} It should be clear
that the results established here readily allow to assess how resources
for multi-partite protocols can be prepared from multi-partite states given in some form. 
In particular, GHZ states readily provide
a resource for quantum secret sharing~\cite{PhysRevA.59.1829,SecretSharing}
in which a message is split into parts so that no subset of parties is able to access
the message, while at the same time the entire set of parties is. 
It also gives rise to an efficient scheme of quantum secret
sharing requiring purely classical communication during the reconstruction
phase \cite{GHZAnne}. 

The significance in the established results
on multi-partite entanglement transformations hence lies in the way
they help understanding how multi-partite resources for protocols beyond point-to-point
schemes in quantum networks 
 can be prepared
and manipulated. We expect this to be particularly important when thinking of applications of
transforming 
resources into the desired form in quantum networks
\cite{QRouting1,QRouting2,QRouting3}: Here, multi-partite entanglement 
is conceived to be created by local processes and bi-partite transmissions involving pairs
of nodes, followed by steps of 
entanglement manipulation, which presumably involve instances of 
classical routing techniques. Hence, we see this work as a 
significant 
contribution to how
a quantum internet \cite{QuantumInternet} can possibly be conceived.

\medskip
\textbf{\emph{Conclusions.}} In this work, we have reported substantial progress on asymptotic state transformation via multipartite local operations and classical communication, tackling an important
long-standing problem which to large extent remained open since the early development of quantitative entanglement theory \cite{PhysRevA.63.012307}.
Similar techniques may also prove helpful in the study of other quantum resource theories different from entanglement, such as the resource theory of quantum coherence~\cite{Streltsov2016} and quantum thermodynamics \cite{Lostaglio2015,Cwiklinski2015}.

Putting notions of entanglement combing into a fresh light, we have
been able to derive stringent bounds on multi-partite entanglement
transformations. This progress seems particularly relevant in the
light of the advent of quantum networks and the quantum internet in
which multi-partite features are directly exploited beyond point-to-point
architectures. It is the hope that the present work stimulates further
progress in the understanding of multi-partite protocols. 

\medskip
\textbf{\emph{Acknowledgements. }} We acknowledge discussions with Pawe\l{} Horodecki and financial support by the Alexander von Humboldt-Foundation, the National  Science Center in Poland (POLONEZ UMO-2016/21/P/ST2/04054), the BMBF (Q.com, Q.Link.X), and the ERC (TAQ). This work was further supported by the "Quantum Coherence and Entanglement for Quantum Technology" project, carried out within the First Team programme of the Foundation for Polish Science co-financed by the European Union under the European Regional Development Fund.

 \bibliographystyle{apsrev4-1}

%

\appendix

\section{\label{sec:ProofLemma} Proof of Lemma \ref{lem:combing}}

The proof presented below will be based on the protocol known as
\emph{entanglement combing}~\cite{Yang2009}. We will review this
protocol for a tripartite state $\psi = \psi^{ABC}$. In this case, entanglement combing
transforms the state $\psi^{ABC}$ into $\mu^{A_{1}B}\otimes\nu^{A_{2}C}$
with pure states $\mu$ and $\nu$. Clearly, the transformation is not possible if any of the inequalities~(\ref{eq:combing-2})
is violated. We will now show the converse, i.e., any pair of pure
states $\mu^{A_{1}B}$ and $\nu^{A_{2}C}$ which fulfill the inequalities~(\ref{eq:combing-2})
can be obtained from $\psi^{ABC}$ via LOCC in the asymptotic limit.
For this, we will distinguish between the following cases.

\textbf{Case 1:} $S(\psi^{A})\geq S(\psi^{B})\geq S(\psi^{C})$. In this case,
Bob can send his part of the state $\psi$ to Alice by applying quantum
state merging~\cite{Horodecki2005,Horodecki2007}. This procedure
is possible by using LOCC operations between Alice and Bob. Additionally,
Alice and Bob gain singlets at rate $S(\psi^{A})-S(\psi^{AB})=S(\psi^{A})-S(\psi^{C})$.
The overall process thus achieves the transformation~(\ref{eq:combing-1})
with 
\begin{equation}
\begin{aligned}E(\mu^{A_{1}B}) & =S(\psi^{A})-S(\psi^{C}),\\
E(\nu^{A_{2}C}) & =S(\psi^{C}).
\end{aligned}
\label{eq:merging-1}
\end{equation}
Alternatively, Charlie can send his part of the state $\psi$ to Alice,
thus gaining singlets at rate $S(\psi^{A})-S(\psi^{B})$. In this
way they achieve the transformation~(\ref{eq:combing-1}) with 
\begin{equation}
\begin{aligned}E(\mu^{A_{1}B}) & =S(\psi^{B}),\\
E(\nu^{A_{2}C}) & =S(\psi^{A})-S(\psi^{B}).
\end{aligned}
\label{eq:merging-2}
\end{equation}

In the next step we apply-time sharing, i.e., the first procedure
is performed with probability $p$ and the second with probability
$(1-p)$. In this way, we see that the transformation~(\ref{eq:combing-1})
is possible for any pair of states $\mu^{A_{1}B}$ and $\nu^{A_{2}C}$
with the property 
\begin{equation}
\begin{aligned}E(\mu^{A_{1}B}) & =p\left(S(\psi^{A})-S(\psi^{C})\right)+(1-p)S(\psi^{B}),\\
E(\nu^{A_{2}C}) & =pS(\psi^{C})+(1-p)\left(S(\psi^{A})-S(\psi^{B})\right).
\end{aligned}
\end{equation}
By using subadditivity of von Neumann entropy it is now straightforward to check that for a suitable choice of $p$,
the quantities $E(\mu^{A_{1}B})$ and $E(\nu^{A_{2}C})$ can attain
any value compatible with conditions \begin{subequations} 
\begin{align}
E(\mu^{A_{1}B})+E(\nu^{A_{2}C}) & =S(\psi^{A}),\\
E(\mu^{A_{1}B}) & \leq S(\psi^{B}),\\
E(\nu^{A_{2}C}) & \leq S(\psi^{C}).
\end{align}
\end{subequations}This completes the proof of Lemma~\ref{lem:combing} for 
Case 1.

\textbf{Case 2:} $S(\psi^{B})\geq S(\psi^{C})\geq S(\psi^{A})$. In this case,
Alice, Bob, and Charlie apply assisted entanglement distillation~\cite{DiVincenzo1999,Smolin2005},
with Charlie being the assisting party. This procedure achieves the
transformation~(\ref{eq:combing-1}) with
\begin{equation}
\begin{aligned}E(\mu^{A_{1}B}) & =\min\left\{ S(\psi^{A}),S(\psi^{B})\right\} =S(\psi^{A}),\\
E(\nu^{A_{2}C}) & =0.
\end{aligned}
\label{eq:assisted}
\end{equation}
Alternatively, they can apply assisted entanglement distillation with
Bob being the assisting party, thus achieving 
\begin{equation}
\begin{aligned}E(\mu^{A_{1}B}) & =0,\\
E(\nu^{A_{2}C}) & =\min\left\{ S(\psi^{A}),S(\psi^{C})\right\} =S(\psi^{A}).
\end{aligned}
\end{equation}
By applying time-sharing, we see that we can achieve the transformation~(\ref{eq:combing-1})
with any states $\mu^{A_{1}B}$ and $\nu^{A_{2}C}$ fulfilling \begin{subequations}
\begin{align}
E(\mu^{A_{1}B}) & =pS(\psi^{A}),\\
E(\nu^{A_{2}B}) & =(1-p)S(\psi^{A}).
\end{align}
\end{subequations} 
This completes the proof of Lemma~\ref{lem:combing} for Case 2.

\textbf{Case 3:} $S(\psi^{B})\geq S(\psi^{A})\geq S(\psi^{C})$. Here, we will
apply a combination of protocols used in Case 1 and 2. In particular,
Bob can send his part of the state $\psi$ to Alice by quantum state
merging, see Eq.~(\ref{eq:merging-1}). Alternatively, they can
apply assisted entanglement distillation, see Eq.~(\ref{eq:assisted}).
By time-sharing we obtain
\begin{equation}
\begin{aligned}E(\mu^{A_{1}B}) & =S(\psi^{A})-pS(\psi^{C}),\\
E(\nu^{A_{2}C}) & =pS(\psi^{C}).
\end{aligned}
\end{equation}
By a suitable choice of the probability $p$ it is now possible to
obtain any pair of states $\mu^{A_{1}B}$ and $\nu^{A_{2}C}$ such
that 
\begin{equation}
\begin{aligned}E(\mu^{A_{1}B})+E(\nu^{A_{2}C}) & =S(\psi^{A}),\\
E(\mu^{A_{1}B}) & \leq S(\psi^{A}),\\
E(\nu^{A_{2}C}) & \leq S(\psi^{C}).
\end{aligned}
\end{equation}
This completes the proof of Lemma~\ref{lem:combing} for Case 3. Note that any other
case can be obtained from the above three cases by interchanging the
role of Bob and Charlie. Thus, the proof of the Lemma is complete.

\section{\label{ap:nconver}Proof of Theorem \ref{thm:bound2}}

{Here, we present the proof of Theorem \ref{thm:bound2}. The ideas presented in the following generalize the proof of Theorem \ref{thm:bound}
for tripartite pure state conversion. In particular, starting with
the $N+1$-partite state $\psi=\psi^{AB_{1}\dots B_{N}}$, we will apply
entanglement combing \cite{Yang2009}
on Alice and all other parties (here referred to as ``all the Bobs''), aiming to get bipartite entanglement between Alice and each of the parties $B_{i}$. If $E_{i}$ denotes
the entanglement between Alice and $i$-th Bob after this procedure, the rate for state conversion from $\psi$ to $\phi=\phi^{AB_{1}\dots B_{N}}$ is bounded below as 
\begin{equation}
R(\psi\rightarrow\phi)\geq\min_{i}\left\{ \frac{E_{i}}{S(\phi^{B_{i}})}\right\} .\label{eq:C-bound1}
\end{equation}
To achieve conversion at rate $\min_{i}\left\{ E_{i}/S(\phi^{B_{i}})\right\} $,
Alice locally prepares the state $\phi^{A\tilde{A}_{1}\dots \tilde{A}_{N}}$,
applies Schumacher compression~\cite{Schumacher1995} to the registers
$\tilde{A}_{i}$, and distributes them among the Bobs by using entanglement
which has been combed in the previous procedure. In the rest of this section, we will show that combing can achieve an $N$-tuple of singlet rates $(E_1,\dots, E_N)$ such that
\begin{align}
\min_{i}\left\{ \frac{E_{i}}{S(\phi^{B_{i}})}\right\}  & \geq m^{\psi,\phi} := \min_{X}\left\{ \frac{S(\psi^{AX})}{\sum_{B_{i}\notin X}S(\phi^{B_{i}})}\right\} ,\label{eq:C-bound2}
\end{align}
where $X$ denotes a subset of all the Bobs. When there is no ambiguity, we will denote $m^{\psi,\phi}$ simply by $m$.

In the first step of the proof we will consider all possible ways
to merge Bobs' parts of the state $B_{i}$ with Alice. Since in the
scenario considered here we have $N$ Bobs, there are $N!$ different
ways to achieve this, depending on the order of the Bobs in the merging
procedure. We will first consider entanglement $N$-tuple
$(E_{1},\dots ,E_{N})$, where $E_{i}$ denotes the amount of entanglement
shared between Alice and $i$-th Bob after the merging procedure.
For example, taking $N=4$, merging first $B_1$, then $B_2$, then $B_3$ and 
finally $B_4$ to Alice will achieve the $4$-tuple: 
\begin{subequations} \begin{align}
    E_1 &= S(\psi^{A})-S(\psi^{AB_1}),\\
    E_2 &= S(\psi^{AB_1}) - S(\psi^{AB_1B_2}),\\
    E_3 &= S(\psi^{AB_1B_2}) - S(\psi^{AB_1B_2B_3}),\\
    E_4 &= S(\psi^{AB_1B_2B_3}),
\end{align} \end{subequations}
while merging first $B_3$, then $B_1$, then $B_4$ and finally $B_2$ to Alice will achieve the $4$-tuple:
\begin{subequations} \begin{align}
    E_1 &= S(\psi^{AB_3})-S(\psi^{AB_1B_3}),\\
    E_2 &= S(\psi^{AB_1B_3B_4}),\\
    E_3 &= S(\psi^{A}) - S(\psi^{AB_3}),\\
    E_4 &= S(\psi^{AB_1B_3}) - S(\psi^{AB_1B_3B_4}).
\end{align} \end{subequations}

The aforementioned $N!$ merging procedures give rise to $N!$ $N$-tuples, which we
will name the "entanglement extreme points". We note that some of the values $E_i$ can be negative, implying that entanglement is consumed in this case. Proposition 2 of Ref.~\cite{Yang2009} guarantees that for any $N$-tuple $(E_1,\dots, E_N)$ with the properties
\begin{enumerate}[(i)]
    \item $\forall i \in \{1,\dots,N\}$, $E_i \geq 0$,
    \item $(E_1,\dots,E_N)$ is in the convex polytope spanned by the entanglement extreme points, 
\end{enumerate}
there exists an asymptotic LOCC protocol acting on the state $\psi$ and distilling singlets between Alice and each of the Bobs $B_i$ at rate $E_i$.
In the following, we are interested in the renormalized entanglement rates
\begin{align} R_i=\frac{E_i}{S(\phi^{B_i})}, \end{align} see also Eq.~(\ref{eq:C-bound1}).
We can define for each $N$-tuple $(E_1,\dots,E_N)$ an $N$-tuple $(R_1,\dots,R_N)$. 
We will consider from now on only the tuples $(R_1,\dots,R_N)$, which will also be called "rate distributions". We will call "extreme points" the rates distribution defined from the entanglement extreme points. It is easily seen from previous combing condition and Eq.~(\ref{eq:C-bound1}) that, if we find a distribution of rates $(R_1,\dots, R_N)$ satisfying
\begin{enumerate}[(i)]
    \item $\forall i \in \{1,\dots,N\}$, $R_i \geq 0$,
    \item $(R_1,\dots,R_N)$ is in the convex polytope spanned by the extreme points,
\end{enumerate}
we will be able to achieve conversion from $\psi$ to $\phi$ with rate
\begin{equation}
    R(\psi\rightarrow\phi)\geq\min_{i}\left\{R_i\right\}.
\end{equation}
In order to prove Eq.~(\ref{eq:C-bound2}), we will find in the convex set of the extreme points
a point $(R_1,\dots,R_N)$ such that
\begin{align}
\min_{i}\left\{R_i\right\}  & \geq\min_{X}\left\{ \frac{S(\psi^{AX})}{\sum_{B_{i}\notin X}S(\phi^{B_{i}})}\right\}.\label{eq:C-bound3}
\end{align}

The outline of the rest of the proof is as follows: in the first step we will construct by convexity a set of points $(R_1, \dots, R_N)$ satisfying $R_N \geq m^{\psi,\phi}$ from the extreme points. We note that the convex set of these newly constructed points will only contain rate distributions with  $N^{\textrm{th}}$ coordinate superior to $m^{\psi,\phi}$. From our constructed points, we will construct by convexity a new set of points $(R_1,\dots,R_N)$ satisfying $R_{N-1} \geq m^{\psi,\phi}$. This will lead to a set of point satisfying both $R_N \geq m^{\psi,\phi}$ and $R_{N-1} \geq m^{\psi,\phi}$. The procedure
will continue with $R_{N-2}$ until $R_1$. In this way, we will achieve a distribution $(R_1,\dots,R_N)$ satisfying $\forall i\in\{1,\dots,N\}, R_i \geq m^{\psi,\phi}$. Such a distribution will ensure conversion from $\psi$ to $\phi$ with a rate of at least $m^{\psi,\phi}$, as claimed. \newline

\textbf{First step.} Each of the extreme points is the result of merging the Bobs to Alice in different order. Thus, we can associate each extreme point to a permutation $\sigma$ on the set $\{1,\dots,N\}$.
We denote the set of all permutations by $\mathcal{S}_N$. Moreover, $\sigma(k)=l$ means that $B_l$ is the $k^\textrm{th}$ Bob merged to Alice. It implies that,
\begin{align}
    R_{\sigma(k)}^{\sigma} = R_l^{\sigma} &= \frac{S(\psi^{AB_{\sigma(1)}\dots B_{\sigma(k-1)}}) - S(\psi^{AB_{\sigma(1)}\dots B_{\sigma(k-1)}B_l})}{S(\phi^{B_l})} \nonumber \\
    &=\frac{S(\psi^{AY^\sigma_{k-1}}) - S(\psi^{AY^{\sigma}_{k-1}B_l})}{S(\phi^{B_l})},
\end{align}
where we used the notation $Y^\sigma_{k}=\{B_{\sigma(1)},\dots,B_{\sigma(k)}\}$.

Our next observation is that we can group the $N!$ extreme points in $(N-1)!$ sets of $N$ points. In the following, we denote by $c_{N-i}$ the permutations defined for $i\in\{0,\dots,N-1\}$ as
\begin{subequations} 
\begin{align}
    c_{N-i}(k) &= k\textrm{, }\forall k\in \{1,\dots,N-i-1\},\\
    c_{N-i}(N-i) &= N,\\
    c_{N-i}(k) &= k - 1\textrm{, }\forall k\in\{N-i+1,\dots,N\}.
\end{align}
\end{subequations}
Consider now a distribution $(R_1^{\sigma},\dots,R_N^{\sigma})$ with $\sigma(N)=N$, i.e., $B_N$ merged in $N^\textrm{th}$
position. We form a set by
grouping together the $N$ distributions $(R_1^{\sigma\circ c_{N-i}},\dots, R_N^{\sigma\circ c_{N-i}})$.
In term of merging order, the distribution $\sigma\circ c_{N-i}$ give rise to
the following ordering:
\begin{enumerate}
    \item For $k < N-i$, $B_{\sigma\circ c_{N-i}(k)}=B_{\sigma(k)}$ is merged in position $k$,
    \item For $k = N-i$, $B_{\sigma\circ c_{N-i}(N-i)}=B_{N}$ is merged in position $N-i$,
    \item For $N \geq k > N-i$, $B_{\sigma\circ c_{N-i}(k)}=B_{\sigma(k-1)}$ is merged in position $k$. 
\end{enumerate}
Distributions $\sigma\circ c_{N-i}$ are the distributions
obtained by merging Bobs $1$ to $N-1$ with the relative order given by $\sigma$. The only difference
is the merging position of $B_N$.

We can order this set by the value of the $N^\textrm{th}$ coordinate. Indeed, 
\begin{equation}
    R_N^\sigma \geq R_N^{\sigma\circ c_{N-1}} \geq R_N^{\sigma\circ c_{N-2}} \geq \dots \geq R_N^{\sigma\circ c_1}. \label{eq:hierarchy}
\end{equation}
Note that $\sigma \circ c_N = \sigma$. For a proof of Eq.~(\ref{eq:hierarchy}) in the general case see Appendix \ref{sec:ProofOrder}.
There are $(N-1)!$ distributions satisfying $\sigma(N)=N$. We have $(N-1)!$ ordered sets of size $N$. Observe that for all $\sigma\in S_N$ satisfying $\sigma(N)=N$,
\begin{equation}
    R_N^{\sigma} = \frac{S(\psi^{AB_1\dots B_{N-1}})}{S(\phi^{B_N})} \in \left\{ \frac{S(\psi^{AX})}{\sum_{B_{i}\notin X}S(\phi^{B_{i}})}\right\}
\end{equation}
As a consequence, 
\begin{equation}
R_N^{\sigma} \geq m^{\psi,\phi}.
\end{equation}
Two situations can happen for each of the $(N-1)!$ sets. The first case is that $R_N^{\sigma\circ c_1}\geq m^{\psi,\phi}$. In this case, we can obtain the distribution $(R_1^{\sigma\circ c_1},\dots,R_{N-1}^{\sigma\circ c_1},m^{\psi,\phi})$
from $(R_1^{\sigma\circ c_1},\dots,R_{N}^{\sigma\circ c_1})$ by simply reducing the entanglement between Alice and $B_N$.

The second case is that we can find $i$ such that $R_N^{\sigma\circ c_{N-i}} \geq m^{\psi,\phi} > R_N^{\sigma\circ c_{N-i-1}}$. 
In this case, we can consider a convex combination of $R^{\sigma\circ c_{N-i}}$ and $R^{\sigma\circ c_{N-i-1}}$, in order to arrive at
a resulting distribution $(R_1,\dots,R_N)$ such that $R_N = m^{\psi,\phi}$. We also know easily the value of most of
the two distribution's coordinates. Indeed, 
\begin{enumerate}
    \item For $k < N-i-1$, $c_{N-i-1}(k)=c_{N-i}(k)=k$, which gives
    \begin{subequations}
    \begin{align}
        R^{\sigma\circ c_{N-i-1}}_{\sigma\circ c_{N-i-1}(k)}=&R^{\sigma\circ c_{N-i-1}}_{\sigma(k)}=\frac{S(\psi^{AY^\sigma_{k-1}}) - S(\psi^{AY^\sigma_k})}{S(\phi^{B_{\sigma(k)}})},\\
        R^{\sigma\circ c_{N-i}}_{\sigma\circ c_{N-i}(k)}=&R^{\sigma\circ c_{N-i}}_{\sigma(k)}=\frac{S(\psi^{AY^\sigma_{k-1}}) - S(\psi^{AY^\sigma_k})}{S(\phi^{B_{\sigma(k)}})}.
    \end{align}
    \end{subequations}
    \item For $ k= N-i-1$, $c_{N-i-1}(N-i-1)=N$ and $c_{N-i}(N-i-1)=N-i-1$, 
    \begin{subequations}
    \begin{align}
        R^{\sigma\circ c_{N-i-1}}_{\sigma\circ c_{N-i-1}(N-i-1)}=&R^{\sigma\circ c_{N-i-1}}_{N}=\frac{S(\psi^{AY_{N-i-2}^\sigma}) - S(\psi^{AY_{N-i-2}^\sigma B_N})}{S(\phi^{B_N})},\\
        R^{\sigma\circ c_{N-i}}_{\sigma\circ c_{N-i}(N-i-1)}=&R^{\sigma\circ c_{N-i}}_{\sigma(N-i-1)}=\frac{S(\psi^{AY_{N-i-2}^\sigma}) - S(\psi^{AY_{N-i-1}^\sigma})}{S(\phi^{B_{\sigma(N-i-1)}})}.
    \end{align}
    \end{subequations}
    \item For $k=N-i$, $c_{N-i-1}(N-i)=N-i-1$ and $c_{N-i}(N-i)=N$, 
    \begin{subequations}
    \begin{align}
        R^{\sigma\circ c_{N-i-1}}_{\sigma\circ c_{N-i-1}(N-i)}=&R^{\sigma\circ c_{N-i-1}}_{\sigma(N-i-1)}=\frac{S(\psi^{AY_{N-i-2}^\sigma B_N}) - S(\psi^{AY_{N-i-1}^\sigma B_N})}{S(\phi^{B_{\sigma(N-i-1)}})},\\
        R^{\sigma\circ c_{N-i}}_{\sigma\circ c_{N-i}(N-i)}=&R^{\sigma\circ c_{N-i}}_{N}=\frac{S(\psi^{AY_{N-i-1}^\sigma}) - S(\psi^{AY_{N-i-1}^\sigma B_N})}{S(\phi^{B_N})}.
    \end{align}
    \end{subequations}
    \item For $k>N-i$, $c_{N-i-1}(k)=c_{N-i}(k)=k-1$,
    \begin{subequations}
    \begin{align}
        R^{\sigma\circ c_{N-i-1}}_{\sigma\circ c_{N-i-1}(k)}=&R^{\sigma\circ c_{N-i-1}}_{\sigma(k-1)}=\frac{S(\psi^{AY_{k-2}^\sigma B_N}) - S(\psi^{AY_{k-1}^\sigma B_N})}{S(\phi^{B_{\sigma(k-1)}})},\\
        R^{\sigma\circ c_{N-i}}_{\sigma\circ c_{N-i}(k)}=&R^{\sigma\circ c_{N-i}}_{\sigma(k-1)}=\frac{S(\psi^{AY_{k-2}^\sigma B_N}) - S(\psi^{AY_{k-1}^\sigma B_N})}{S(\phi^{B_{\sigma(k-1)}})}.
    \end{align}
    \end{subequations}
\end{enumerate}
Only two coordinates differ in the distributions given by $\sigma\circ c_{N-i}$ and $\sigma\circ c_{N-i-1}$.
As a consequence, the distribution resulting from their convex combination will be 
a distribution with $N^{\textrm{th}}$ coordinate taking
the value $m^{\psi,\phi}$, while the
$\sigma(N-i-1)^{\textrm{th}}$ one assumes the value
\begin{equation}
    \frac{S(\psi^{AY_{N-i-2}^\sigma})-m^{\psi,\phi}S(\phi^{B_N}) - S(\psi^{AY_{N-i-1}^\sigma B_N})}{S(\phi^{B_{\sigma(N-i-1)}})},
\end{equation} 
and $\forall k\in\{1,\dots,N-1\}\setminus\{N-i-1\}$, the $k^{\textrm{th}}$ coordinate take the 
value $R_{\sigma(k)}^{\sigma\circ c_{N-i}}$.

We will apply this procedure for each $\sigma\in \mathcal{S}_N$ with $\sigma(N)=N$.
We associate the resulting distributions with the $\sigma$ that gave rise to the distribution we used in the convex combination.
The result are $(N-1)!$ distributions $(R_1^{\sigma},\dots,R_{N-1}^\sigma,m^{\psi,\phi})$ one
for each permutation $\sigma$.
For the given quantum state $\psi$ 
equipped with the partitioning
in $A$ and $\{B_1,\dots,B_{N-1}\}$, we now
define the function 
\begin{equation}
        S_2^\psi: X\subset \{B_1,\dots,B_{N-1}\}
        \rightarrow \mathbb{R}^+_0
    \end{equation}
that depends on subsets
$X\subset \{B_1,\dots,B_{N-1}\}$, taking the values
\begin{equation}
S_2^\psi(X) := 
\left\{
\begin{array}{l}
S(\psi^{AX}) - m^{\psi,\phi}S(\phi^{B_N}), \\
\,\,\,\,\,\, \text{ if }  
\frac{S(\psi^{AX})-S(\psi^{AXB_N})}{S(\phi^{B_N})} < m(\psi, \phi),\\
S(\psi^{AXB_N}), \\
\,\,\,\,\,\,\text{ if }  
\frac{S(\psi^{AX})-S(\psi^{AXB_N})}{S(\phi^{B_N})} \geq m^{\psi,\phi}.
\end{array}
\right.
\end{equation}
We can rewrite the coordinates of $(R_1^{\sigma},\dots,R_{N-1}^\sigma, m^{\psi,\phi})$
as a function of $S_2^\psi$.
\begin{enumerate}
    \item For $k < N-i-1$, $R_N^{\sigma\circ c_k} < R_N^{\sigma\circ c_{k+1}} \leq R_N^{\sigma\circ c_{N-i-1}} < m^{\psi,\phi}$.
    As a consequence,
     \begin{align}
        R^{\sigma}_{\sigma(k)}&=\frac{S(\psi^{AY^\sigma_{k-1}}) - S(\psi^{AY^\sigma_{k}})}{S(\phi^{B_{\sigma(k)}})}\\
        &=\frac{S_2^\psi(Y^\sigma_{k-1}) - S_2^\psi(Y^\sigma_{k})}{S(\phi^{B_{\sigma(k)}})}.
    \end{align}
    
    \item For $k = N-i-1$, $R_N^{\sigma\circ c_{N-i-1}}<m^{\psi,\phi}\leq R_N^{\sigma\circ c_{N-i}}$,
      \begin{align}
        R^\sigma_{\sigma(N-i-1)}&=\frac{S(\psi^{AY^\sigma_{N-i-2}})-m^{\psi,\phi}S(\phi^{B_N}) - S(\psi^{AY^\sigma_{N-i-1}B_N})}{S(\phi^{B_{\sigma(N-i-1)}})}\\
        &=\frac{S_2^\psi(Y^\sigma_{k-1}) - S_2^\psi(Y^\sigma_{k})}{S(\phi^{B_{\sigma(k)}})}.
    \end{align}
    
    \item For $N > k > N-i-1$, $m^{\psi,\phi} \leq R_N^{\sigma\circ c_{N-i}} \leq R_N^{\sigma\circ c_{k}} \leq R_N^{\sigma\circ c_{k+1}}$,
     \begin{align}
        R^{\sigma}_{\sigma(k)}&=\frac{S(\psi^{AY^\sigma_{k-1}B_N}) - S(\psi^{AY^\sigma_{k}B_N})}{S(\phi^{B_{\sigma(k)}})}\\
        &=\frac{S_2^\psi(Y^\sigma_{k-1}) - S_2^\psi(Y^\sigma_{k})}{S(\phi^{B_{\sigma(k)}})}
        .
    \end{align}
    
\end{enumerate}
In summary, the have 
just presented first step of the procedure leaves us with $(N-1)!$ distributions $(R_1^\sigma,\ldots,R_{N-1}^\sigma,m^{\psi,\phi})$.\newline 

We introduce now generalized functions which will be used in the following steps. We define in a recursive way the functions $S_{j}^\psi$ for $j\in\{1,\dots, N\}$ by
\begin{subequations}
    \begin{equation}
        S_{j}^\psi: X\subset \{B_1,\dots,B_{N-j+1}\}
        \rightarrow \mathbb{R}^+_0,
\end{equation}
\begin{equation}
        S_{1}^\psi(X) := S(\psi^{AX}),
\end{equation}
\begin{equation}
S_{j+1}^\psi(X) := 
\left\{
\begin{array}{l}
S_{j}^\psi(X) - m_{j}^{\psi,\phi}S(\phi^{B_{N-j+1}}), \\
\,\,\,\,\,\, \text{ if }  
\frac{S_{j}^\psi(X)-S_{j}^\psi(XB_{N-j+1})}{S(\phi^{B_{N-j+1}})} < m_{j}^{\psi, \phi},\\
S_{j}^\psi(XB_{N-j+1}), \\
\,\,\,\,\,\,\text{ if }  
\frac{S_{j}^\psi(X)-S_{j}^\psi(XB_{N-j+1})}{S(\phi^{B_{N-j+1}})} \geq m^{\psi,\phi}_{j}. \label{def:sj_plus_2}
\end{array}
\right.
\end{equation}
\end{subequations}
Moreover, $m_j^{\psi,\phi}$ is given as follows,
\begin{equation}
    m_{j}^{\psi,\phi} := \min\left\{\frac{S_{j}^\psi(X)}{\sum_{B_i\notin X} S(\phi^{B_i})}, X\subset\{B_1,\dots,B_{N-j+1}\}\right\}.
\end{equation}
We show in Appendix (\ref{eq:strong_sub_annex}) that all the function $S_{j}$ satisfy strong subadditivity on the subsets of Bobs such that $\forall X\subset\{B_1,\dots,B_{N-j+1}\}$ and for $B_l,B_m \notin X$,
\begin{equation}
    S^\psi_{j}(XB_l) + 
    S_{j}^\psi(XB_m) 
    \geq S^\psi_{j}(XB_lB_m) + 
    S_{j}^\psi(X).\label{eq:strong_sub_j}
\end{equation}
Equipped with these tools, we are now ready to present the general $(j+1)^\mathrm{th}$ step of the procedure, where we will make extensive use of the properties of $R_i^\sigma$ and the generalized functions $S_j^\psi$ and $m_j^{\psi,\phi}$ discussed above. \newline

\textbf{$(j+1)^{\textrm{th}}$ step.}
In the $(j+1)^{\textrm{th}}$ step, there are $(N-j)!$ 
distributions denoted as
$(R_1^\sigma,\dots,R^\sigma_{N-j},m_j^{\psi,\phi}, m_{j-1}^{\psi,\phi},\dots, m^{\psi,\phi})$. One for each
$\sigma\in\mathcal{S}_N$ with $\forall k\in\{N-j+1,\dots,N\}$, $\sigma(k)=k$. For $k\in\{1,\dots,N-j\}$, the 
coordinate's values are given by
\begin{equation}
    R_{\sigma(k)}^\sigma = \frac{S_{j+1}(\psi^{AY_{k-1}^\sigma}) - S_{j+1}(\psi^{AY_{k}^\sigma})}{S(\phi^{B_{\sigma(k)}})},
\end{equation}

We will construct by convexity $(N-j-1)!$ distributions $(R_1,\dots,R_{N-j-1},m_{j+1}^{\psi,\phi},\dots,m^{\psi,\phi})$.
We proceed as before and group 
distributions in $(N-j-1)!$ sets of $N-j$ distributions.
We consider distributions associated with permutations $\sigma$ verifying $\sigma(N-j)=N-j$.
For $i\in\{0,\dots,N-j-1\}$, we define the permutations,
\begin{subequations}
\begin{align}
    c^{j+1}_{N-j-i}(k) &= k\textrm{, }\forall k\in \{1,\dots,N-j-i-1\},\\
    c^{j+1}_{N-j-i}(N-j-i) &= N-j,\\
    c^{j+1}_{N-j-i}(k) &= k - 1\textrm{, }\forall k\in\{N-j-i+1,\dots,N-j\},
\end{align}
\label{eq:c_Nj_def}
\end{subequations}
and we group the distributions 
$R^{\sigma\circ c^{j+1}_{N-j-i}}$.
For the sake of clarity, we drop the superscript of the $c$ permutations and we write
$N_j:=N-j$ for the rest of the proof.
We arrive at a 
hierarchy in the coordinates $N_j$, i.e. (see Appendix \ref{sec:ProofOrder}),
\begin{equation}
    R_{N_j}^{\sigma\circ c_{N_j}} \geq R_{N_j}^{\sigma\circ c_{N_j-1}}\geq \dots\geq R_{N_j}^{\sigma\circ c_{1}} \label{eq:R_Nj_hierarchy}
\end{equation}
with
\begin{equation}
R^{\sigma\circ c_{N_j}}_{N_j} \in \left\{\frac{S_{j+1}(X)}{\sum_{B_i\notin X}S(\phi^{B_i})}\right\}.
\end{equation}
As a consequence,
\begin{equation}
R^{\sigma\circ c_{N_j}}_{N_j} \geq m_{j+1}. 
\end{equation}
As in the first step, if $R^{\sigma\circ c_{1}}_{N_j} \geq m_{j+1}$, then we can take the distributions $R^{\sigma\circ c_1}$ and reduce entanglement
to achieve a distribution $(R^f_{1},\dots,R^{f}_{N_j-1},m_{j+1},\dots,m)$. Else, we can find an $i$ such that 
\begin{equation}
R^{\sigma\circ c_{N_j-i}}_{N_j} \geq m_{j+1} > R^{\sigma\circ c_{N_j-i-1}}_{N_j}.
\end{equation}
Again following the same ideas as in the first step, we take a convex combination of the two distributions $R^{\sigma\circ c_{N_j-i}}$ and $R^{\sigma\circ c_{N_j-i-1}}$. The values of all coordinates are given by
\begin{enumerate}
    \item For $k < N-i-j-1$, $c_{N_j-i-1}(k)=c_{N_j-i}(k)=k$, we obtain
    \begin{align}
        R^{\sigma\circ c_{N_j-i-1}}_{\sigma\circ c_{N_j-i-1}(k)}=&R^{\sigma\circ c_{N_j-i-1}}_{\sigma(k)}=\frac{S_{j+1}(\psi^{AY^\sigma_{k-1}}) - S_{j+1}(\psi^{AY^\sigma_k})}{S_{j+1}(\phi^{B_{\sigma(k)}})},\\
        R^{\sigma\circ c_{N_j-i}}_{\sigma\circ c_{N_j-i}(k)}=&R^{\sigma\circ c_{N_j-i}}_{\sigma(k)}=\frac{S_{j+1}(\psi^{AY^\sigma_{k-1}}) - S_{j+1}(\psi^{AY^\sigma_k})}{S(\phi^{B_{\sigma(k)}})}.
    \end{align}
    \item For $ k= N_j-i-1$, $c_{N_j-i-1}(N_j-i-1)=N$ and $c_{N_j-i}(N_j-i-1)=N_j-i-1$, we obtain
    \begin{align}
        R^{\sigma\circ c_{N_j-i-1}}_{\sigma\circ c_{N_j-i-1}(N_j-i-1)}=&R^{\sigma\circ c_{N_j-i-1}}_{N_j} \\
        =&\frac{S_{j+1}(\psi^{AY_{N_j-i-2}^\sigma}) - S_{j+1}(\psi^{AY_{N_j-i-2}^\sigma B_{N_j}})}{S(\phi^{B_{N_j}})}, \nonumber \\
        R^{\sigma\circ c_{N_j-i}}_{\sigma\circ c_{N_j-i}(N_j-i-1)}=&R^{\sigma\circ c_{N_j-i}}_{\sigma(N_j-i-1)} \\
        =&\frac{S_{j+1}(\psi^{AY_{N_j-i-2}^\sigma}) - S_{j+1}(\psi^{AY_{N_j-i-1}^\sigma})}{S(\phi^{B_{\sigma(N_j-i-1)}})}. \nonumber
    \end{align}
    
    \item For $k=N_j-i$, $c_{N_j-i-1}(N_j-i)=N_j-i-1$ and $c_{N_j-i}(N_j-i)=N_j$, we obtain
    \begin{align}
        R^{\sigma\circ c_{N_j-i-1}}_{\sigma\circ c_{N_j-i-1}(N_j-i)}=&R^{\sigma\circ c_{N_j-i-1}}_{\sigma(N_j-i-1)} \\
        =&\frac{S_{j+1}(\psi^{AY_{N_j-i-2}^\sigma B_{N_j}}) - S_{j+1}(\psi^{AY_{N_j-i-1}^\sigma B_{N_j}})}{S(\phi^{B_{\sigma(N_j-i-1)}})}, \nonumber \\
        R^{\sigma\circ c_{N_j-i}}_{\sigma\circ c_{N_j-i}(N_j-i)}=&R^{\sigma\circ c_{N_j-i}}_{N_j} \\
        =&\frac{S_{j+1}(\psi^{AY_{N_j-i-1}^\sigma}) - S_{j+1}(\psi^{AY_{N_j-i-1}^\sigma B_{N_j}})}{S(\phi^{B_{N_j}})}. \nonumber
    \end{align}
    
    \item For $k>N_j-i$, $c_{N_j-i-1}(k)=c_{N_j-i}(k)=k-1$, we obtain
    \begin{align}
        R^{\sigma\circ c_{N_j-i-1}}_{\sigma\circ c_{N_j-i-1}(k)}=&R^{\sigma\circ c_{N_j-i-1}}_{\sigma(k-1)}=\frac{S_{j+1}(\psi^{AY_{k-2}^\sigma B_{N_j}}) - S_{j+1}(\psi^{AY_{k-1}^\sigma B_{N_j}})}{S(\phi^{B_{\sigma(k-1)}})},\\
        R^{\sigma\circ c_{N_j-i}}_{\sigma\circ c_{N_j-i}(k)}=&R^{\sigma\circ c_{N_j-i}}_{\sigma(k-1)}=\frac{S_{j+1}(\psi^{AY_{k-2}^\sigma B_{N_j}}) - S_{j+1}(\psi^{AY_{k-1}^\sigma B_{N_j}})}{S(\phi^{B_{\sigma(k-1)}})}.
    \end{align}
\end{enumerate}
Again, only two coordinates differ between the distributions given by $\sigma\circ c_{N_j-i}$ and $\sigma\circ c_{N_j-i-1}$.
As a consequence, the distribution resulting from their convex combination will be
a distribution with a $N_j^{\textrm{th}}$ coordinate of value $m_{j+1}^{\psi,\phi}$, a $\sigma(N_j-i-1)^{\textrm{th}}$ coordinate of value
\begin{equation}
    \frac{S_{j+1}^\psi(Y_{N_j-i-2}^\sigma)-m_{j+1}^{\psi,\phi}S(\phi^{B_{N_j}}) - S_{j+1}^\psi(Y_{N_j-i-1}^\sigma B_{N_j})}{S(\phi^{B_{\sigma(N_j-i-1)}})},
\end{equation} 
and $\forall k\in\{1,\dots,N_j-1\}\setminus\{N_j-i-1\}$, a $k^{\textrm{th}}$ coordinate of value $R_{\sigma(k)}^{\sigma\circ c_{N_j-i}}$.
As in the first step, from each permutation $\sigma\in \mathcal{S}_N$ with $\forall k\in\{N-j,\dots,N\}, \sigma(k)=k$ we have a resulting distribution $(R_1^\sigma,\dots,R_{N_j-1}^\sigma, m_{j+1}^{\psi,\phi}, \dots,  m^{\psi,\phi})$ that we label with $\sigma$. All the coordinate $R_{\sigma(k)}^\sigma$ can be rewritten in term of $S_{j+2}^\psi$ such that
\begin{equation}
    R_{\sigma(k)}^\sigma = \frac{S_{j+2}^\psi(Y_{k-1}^\sigma) - S_{j+2}^\psi(Y_{k}^\sigma)}{S(\phi^{B_{\sigma(k)}})}.
\end{equation}

Following this procedure until step $N$, we find ourselves with the distribution $(m_N^{\psi,\phi},m_{N-1}^{\psi,\phi},\dots,m_2^{\psi,\phi},m^{\psi,\phi})$. It remains to be proven that $\forall j\in\{1,\dots,N-1\}$, $m_{j+1}^{\psi,\phi}\geq m_{j}^{\psi,\phi}$.
Taking an element of the set from which $m_{j+1}^{\psi,\phi}$ is the minimum: ${S^\psi_{j+1}(X)}/({\sum_{B_i\notin X} S(\phi^{B_i})})$, 
where $X$ is a subset of $\{B_1,\dots,B_{N_j}\}$, 
we will show it is greater or equal to every elements of the set from which $m_j^{\psi,\phi}$ is the minimum,
\begin{equation}
    M_j^{\psi,\phi}:= \left\{\frac{S_{j}^\psi(Y)}{\sum_{B_i\notin Y} S(\phi^{B_i})}, Y\subset\{B_1,\dots,B_{N_j+1}\}\right\}.
\end{equation}
There are two cases:
\begin{enumerate}
    \item If $S^\psi_{j+1}(X) = S^\psi_{j}(XB_{N_j+1})$, then
    \begin{equation}
        \frac{S^\psi_{j+1}(X)}{\sum_{B_i\notin X} S(\phi^{B_i})} = \frac{S^\psi_{j}(XB_{N_j+1})}{\sum_{B_i\notin Y} S(\phi^{B_i})} \in M_j^{\psi,\phi}.
    \end{equation}
    As a consequence,
      \begin{equation}
    \frac{S^\psi_{j+1}(X)}{\sum_{B_i\notin Y} S(\phi^{B_i})} \geq m_j^{\psi,\phi}.
        \end{equation}
    \item If $S^\psi_{j+1}(X) = S^\psi_{j}(X) - m_j^{\psi,\phi}S(\phi^{B_{N_j+1}})$, we know that
    \begin{equation}
        \frac{S^\psi_j(X)}{\sum_{B_i\notin X} S(\phi^{B_i}) + S(\phi^{B_{N_j+1}})}\geq m_j^{\psi,\phi}.
    \end{equation}
    It implies directly that
    \begin{equation}
        \frac{S^\psi_{j}(X) - m_j^{\psi,\phi}S(\phi^{B_{N_j+1}})}{\sum_{B_i\notin X} S(\phi^{B_i})}\geq m_j^{\psi,\phi}.
    \end{equation}
\end{enumerate}
Thus, recalling that via LOCC it is always possible to reduce bipartite entanglement between Alice and the Bobs, we can finally achieve the distribution $(m^{\psi,\phi},\dots, m^{\psi,\phi})$, and the proof of Theorem~\ref{thm:bound2} is complete.}

{

\section{\label{sec:ProofOrder}Proof of Eqs.~(\ref{eq:hierarchy}) and (\ref{eq:R_Nj_hierarchy})}

To prove Eq.~(\ref{eq:R_Nj_hierarchy}) we will show that $\forall i\in\{0,\dots,N_j-2\}$, 
$R_{N_j}^{\sigma\circ c_{N_j-i}} \geq R_N^{\sigma\circ c_{N_j-i-1}}$.
First, we need to remark that according to definition (\ref{eq:c_Nj_def}),
\begin{align*}
    Y^{\sigma\circ c_{N_j-i}}_{N_j-i-1} &= \{B_{\sigma\circ c_{N_j-i}(1)},\dots, B_{\sigma\circ c_{N_j-i}(N_j-i-1)}\}\\
                                        &= \{B_{\sigma(1)},\dots, B_{\sigma(N_j-i-1)}\}\\
                                        &= Y^{\sigma}_{N_j-i-1}.
\end{align*}
Then rewriting explicitly the coordinates $R^{\sigma\circ c_{N_j-i}}_{N_j}$ and $R^{\sigma\circ c_{N_j-i-1}}_{N_j}$ we obtain
\begin{align}
    R_{N_j}^{\sigma\circ c_{N_j-i}} &= \frac{S_{j+1}^{\psi}(Y^{\sigma\circ c_{N_j-i}}_{N_j-i-1}) - S_{j+1}^{\psi}(Y^{\sigma\circ c_{N_j-i}}_{N_j-i-1}B_{N_j})}{S(\phi^{B_{N_j}})}\\
    &=\frac{S_{j+1}^{\psi}(Y^{\sigma}_{N_j-i-1}) - S_{j+1}^{\psi}(Y^{\sigma}_{N_j-i-1}B_{N_j})}{S(\phi^{B_{N_j}})} \nonumber \\
       &=\frac{S_{j+1}^{\psi}(Y^{\sigma}_{N_j-i-2}B_{\sigma(N_j-i-1)}) - S_{j+1}^{\psi}(Y^{\sigma}_{N_j-i-2}B_{\sigma(N_j-i-1)}B_{N_j})}{S(\phi^{B_{N_j}})}, \nonumber
       \end{align}
and
\begin{align}
    R_{N_j}^{\sigma\circ c_{N_j-i-1}} &= \frac{S_{j+1}^{\psi}(Y^{\sigma\circ c_{N_j-i-1}}_{N_j-i-2}) - S_{j+1}^{\psi}(Y^{\sigma\circ c_{N_j-i-1}}_{N_j-i-2}B_{N_j})}{S(\phi^{B_{N_j}})}\\
    &=\frac{S_{j+1}^{\psi}(Y^{\sigma}_{N_j-i-2}) - S_{j+1}^{\psi}(Y^{\sigma}_{N_j-i-2}B_{N_j})}{S(\phi^{B_{N_j}})}. \nonumber
\end{align}
The ``strong subadditivity'' of Eq.~(\ref{eq:strong_sub_j}) ensures that for all subsets $Y$,
\begin{multline}
    S_{j+1}^{\psi}(YB_{\sigma(N_j-i-1)}) + S_{j+1}^{\psi}(YB_{N_j}) \geq\\ S_{j+1}^{\psi}(YB_{\sigma(N_j-i-1)}B_{N_j})+ S_{j+1}^{\psi}(Y).\label{eq:strong_sub_annex}
\end{multline}
Eq.~(\ref{eq:R_Nj_hierarchy}) follows directly from it, 
since Eq.~(\ref{eq:strong_sub_annex}) implies that
\begin{multline}
    S_{j+1}^{\psi}(YB_{\sigma(N_j-i-1)}) - S_{j+1}^{\psi}(YB_{\sigma(N_j-i-1)}B_{N_j})  \geq\\ S_{j+1}^{\psi}(Y) - S_{j+1}^{\psi}(YB_{N_j}).
\end{multline}
It follows that $R_{N_j}^{\sigma\circ c_{N_j-i}} \geq R_N^{\sigma\circ c_{N_j-i-1}}$. Eqs.~(\ref{eq:hierarchy}) 
are proven in the same manner.
}
\medskip

\section{Proof of Eq.~(\ref{eq:strong_sub_j})} \label{proof_of_strong_sub}
Given that $S_j$ satisfy strong subadditivity, we will show that $\forall X\subset\{B_1,\dots,B_{N-j}\}$ and for $B_l,B_m \notin X$,
\begin{equation}
    S_{j+1}^{\psi}(XB_l) + S_{j+1}^{\psi}(XB_m) - S_{j+1}^\psi(X) - S_{j+1}^\psi(XB_lB_m) \geq 0, \label{eq:ineq_sj_plus_2}
\end{equation}
with $S_{j+1}^\psi$ defined as in Eq.~(\ref{def:sj_plus_2}).

For a given $X\subset\{B_1,\dots,B_{N-j}\}$ and given $B_l,B_m \notin X$, each term of the inequality (\ref{eq:ineq_sj_plus_2}) can be rewritten using $S_{j}^\psi$. For all $Y\subset\{B_1,\dots,B_{N-j}\}$, the value of $S_{j+1}^\psi(Y)$ depends on the value of $S_{j}^\psi(Y) - S_{j}^\psi(YB_{N_j+1})$. As a consequence, several cases arise depending on the value of the four following values,
\begin{subequations}
    \begin{align}
        A&:=\frac{S_{j}^\psi(X)-S_{j}^\psi(XB_{N_j+1})}{S(\phi^{B_{N_j+1}})},\\
        B&:=\frac{S_{j}^\psi(XB_l)-S_{j}^\psi(XB_lB_{N_j+1})}{S(\phi^{B_{N_j+1}})},\\
        C&:=\frac{S_{j}^\psi(XB_m)-S_{j}^\psi(XB_mB_{N_j+1})}{S(\phi^{B_{N_j+1}})},\\
        D&:=\frac{S_{j}^\psi(XB_lB_m)-S_{j}^\psi(XB_lB_mB_{N_j+1})}{S(\phi^{B_{N_j+1}})}.
    \end{align}
\end{subequations}
From Eq.~(\ref{eq:strong_sub_j}), we can deduce $A\leq B$, $A\leq C$, $B\leq D$ and $C \leq D$. We can assume without loss of generality that $B \leq C$. Thus,
\begin{equation}
    A\leq B\leq C\leq D
\end{equation}
and there is only five cases to examine $m_{j}^{\psi,\phi} < A$, $A \leq m_{j}^{\psi,\phi} < B$, $ B \leq m_{j}^{\psi,\phi} < C$, $C \leq m_{j}^{\psi,\phi} < D$ and $D\leq m_{j}^{\psi,\phi}$. We will prove inequality (\ref{eq:ineq_sj_plus_2}) for each of these case.
\begin{enumerate}
    \item $m_{j}^{\psi,\phi} < A$.\newline
    We can rewrite the left-hand side of inequality (\ref{eq:ineq_sj_plus_2}) as
    \begin{multline*}
        S_{j+1}^{\psi}(XB_l) + S_{j+1}^{\psi}(XB_m) - S_{j+1}^\psi(X) - S_{j+1}^\psi(XB_lB_m) =\\ S_{j}^{\psi}(XB_lB_{N_j+1}) + S_{j}^{\psi}(XB_mB_{N_j+1}) - S_{j}^\psi(XB_{N_j+1}) - S_{j}^\psi(XB_lB_mB_{N_j+1}).
    \end{multline*}
    According to Eq.~(\ref{eq:strong_sub_j}),
    \begin{multline*}
        S_{j}^{\psi}(XB_lB_{N_j+1}) + S_{j}^{\psi}(XB_mB_{N_j+1}) -\\ S_{j}^\psi(XB_{N_j+1}) - S_{j}^\psi(XB_lB_mB_{N_j+1}) \geq 0.
    \end{multline*}
    So the inequality is verified.
    \item $A \leq m_{j}^{\psi,\phi} < B$.\newline
    We can rewrite the left-hand side of inequality (\ref{eq:ineq_sj_plus_2}) as
        \begin{multline*}
        S_{j+1}^{\psi}(XB_l) + S_{j+1}^{\psi}(XB_m) - S_{j+1}^\psi(X) - S_{j+1}^\psi(XB_lB_m) =\\ S_{j}^{\psi}(XB_lB_{N_j+1}) + S_{j}^{\psi}(XB_mB_{N_j+1}) - S_{j}^\psi(X) +\\ m_{j}^{\psi,\phi}S(\phi^{B_{N_j+1}})- S_{j}^\psi(XB_lB_mB_{N_j+1}).
    \end{multline*}
    According to Eq.~(\ref{eq:strong_sub_j}), we know that the last equation's right side is larger than
    \begin{equation*}
        S_{j}^\psi(XB_{N_j+1}) - S_{j}^\psi(X) + m_{j}^{\psi,\phi}S(\phi^{B_{N_j+1}}).
    \end{equation*}
    The latter quantity is non-negative because $m_{j}^{\psi,\phi} \geq A$, showing the validity of the inequality.
    \item $B \leq m_{j}^{\psi,\phi} < C$.\newline
     Once again, we rewrite the left-hand side of the inequality (\ref{eq:ineq_sj_plus_2}):
        \begin{multline*}
        S_{j+1}^{\psi}(XB_l) + S_{j+1}^{\psi}(XB_m) - S_{j+1}^\psi(X) - S_{j+1}^\psi(XB_lB_m) =\\ S_{j}^{\psi}(XB_l) + S_{j}^{\psi}(XB_mB_{N_j+1}) - S_{j}^\psi(X) - S_{j}^\psi(XB_lB_mB_{N_j+1}).
        \end{multline*}
        The ``strong subbaditivity'' of the function $S_{j}^\psi$ gives rise to 
        \begin{multline*}
            S_{j}^{\psi}(XB_mB_{N_j+1}) -S_{j}^\psi(XB_lB_mB_{N_j+1})\geq \\
            S_{j}^\psi(XB_{N_j+1}) - S_{j}^\psi(XB_lB_{N_j+1}),
        \end{multline*}
        and this implies that 
        \begin{multline*}
            S_{j}^{\psi}(XB_l) + S_{j}^{\psi}(XB_mB_{N_j+1}) - S_{j}^\psi(X) - S_{j}^\psi(XB_lB_mB_{N_j+1}) \geq \\
             S_{j}^{\psi}(XB_l) + S_{j}^{\psi}(XB_{N_j+1}) - S_{j}^\psi(X) - S_{j}^\psi(XB_lB_{N_j+1}).
        \end{multline*}
        Again, the ``strong subbaditivity'' of $S_{j}^\psi$ allow us to conclude that the right-hand side is positive. Thus, inequality (\ref{eq:ineq_sj_plus_2}) is verified.
        \item $C \leq m_{j}^{\psi,\phi} < D$.\newline
        In this case, the rewriting gives,
        \begin{multline*}
            S_{j+1}^{\psi}(XB_l) + S_{j+1}^{\psi}(XB_m) - S_{j+1}^\psi(X) - S_{j+1}^\psi(XB_lB_m) =\\ S_{j}^{\psi}(XB_l) + S_{j}^{\psi}(XB_m) - m_{j}^{\psi,\phi}S(\phi^{B_{N_j+1}}) -\\ S_{j}^\psi(X) - S_{j}^\psi(XB_lB_mB_{N_j+1}).
        \end{multline*}
        $D$ being superior to $m_{j}^{\psi,\phi}$ implies directly that
        \begin{equation*}
            - S_{j}^\psi(XB_lB_mB_{N_j+1}) - m_{j}^{\psi,\phi}S(\phi^{B_{N_j+1}}) > -S_{j}^\psi(XB_lB_m).
        \end{equation*}
        We can lower bound the right-hand side by
        \begin{equation*}
            S_{j}^{\psi}(XB_l) + S_{j}^{\psi}(XB_m) - S_{j}^\psi(X) - S_{j}^\psi(XB_lB_m).
        \end{equation*}
        Once again, the ``strong subbaditivity'' of $S_{j}^{\psi}$ allows to conclude that the inequality (\ref{eq:ineq_sj_plus_2}) is true.
        \item $D < m_{j}^{\psi,\phi}$. \newline
        The last case is straightforward since the rewriting in term of $S_{j}^\psi$ is
        \begin{multline*}
            S_{j+1}^{\psi}(XB_l) + S_{j+1}^{\psi}(XB_m) - S_{j+1}^\psi(X) - S_{j+1}^\psi(XB_lB_m) =\\ S_{j}^{\psi}(XB_l) + S_{j}^{\psi}(XB_m) - S_{j}^\psi(X) - S_{j}^\psi(XB_lB_m).
        \end{multline*}
        In this case, the ``strong subbaditivity'' of Eq.~(\ref{eq:strong_sub_j}) leads us directly to the conclusion that the inequality~(\ref{eq:ineq_sj_plus_2}) is true.
\end{enumerate}
In conclusion, the inequality (\ref{eq:ineq_sj_plus_2}) is verified for each possible case. Thus Eq.~(\ref{eq:strong_sub_j}) is verified by induction.\\

\section{\label{sec:GHZ}Multi-partite state creation from GHZ states}

In this section, we will show that any $N$-partite mixed state $\sigma=\sigma^{ABC\ldots Z}$
can be obtained from the GHZ state vector $\ket{\mathrm{GHZ}}=(\ket{0}^{\otimes N}+\ket{1}^{\otimes N})/\sqrt{2}$
via asymptotic $N$-partite LOCC at a rate bounded below as 
\begin{equation}
\begin{split}R(\ket{\mathrm{GHZ}}\bra{\mathrm{GHZ}}\rightarrow & \sigma)\geq\\
 & \frac{1}{E_{\mathrm{c}}^{A|BC\ldots Z}(\sigma)+S(\sigma^{C})+\cdots+S(\sigma^{Z})},\label{eq:RMulti-partite}
\end{split}
\end{equation}
where $E_{\mathrm{c}}^{A|BC\ldots Z}$ denotes the entanglement cost
between Alice and the remaining $N-1$ parties. For proving this statement,
we first apply entanglement combing to the $N$-partite GHZ state,
i.e., the asymptotic transformation 
\begin{equation}
\frac{1}{\sqrt{2}}(\ket{0}^{\otimes N}+\ket{1}^{\otimes N})\rightarrow\mu_{1}^{A_{1}B}\otimes\mu_{2}^{A_{2}C}\otimes\mu_{3}^{A_{3}D}\otimes\cdots
\end{equation}
with $N$ pure states $\mu_{i}$. A necessary and sufficient condition
for this transformation is that 
\begin{equation}
\sum_{i}E(\mu_{i})\leq1,\label{eq:EMulti-partite}
\end{equation}
as can be seen by applying multi-partite assisted entanglement distillation~\cite{Smolin2005,Horodecki2005,Horodecki2007}
and time-sharing. The combing is now performed in such a way that
the following equalities hold for some parameter $r\geq0$:\begin{subequations}\label{eq:mu}
\begin{align}
E(\mu_{1}^{A_{1}B}) & =rE_{\mathrm{c}}^{A|BC\ldots Z}(\sigma^{ABC\ldots Z}),\\
E(\mu_{2}^{A_{2}C}) & =rS(\sigma^{C}),\\
 & \vdots\\
E(\mu_{N-1}^{A_{N-1}Z}) & =rS(\sigma^{Z}).
\end{align}
\end{subequations} The parameter $r$ will be determined below.

After combing, Alice and Bob use their state $\mu_{1}^{A_{1}B}$ for
creating the desired final state $\sigma$ via bipartite LOCC. The
optimal rate for this procedure is $E(\mu_{1}^{A_{1}B})/E_{\mathrm{c}}^{A|BC\ldots Z}(\sigma)$,
which is equal to our parameter $r$ due to Eqs.~(\ref{eq:mu}).
In the next step, Bob applies Schumacher compression to those subsystems
of $\sigma$ which are in his possession. The overall compression
rate per copy of the initial state vector $\ket{\mathrm{GHZ}}$ is
given as $r\cdot S(\sigma^{X})$, where $X$ is the corresponding
subsystem. In a final step, Bob teleports compressed parts of the
state $\sigma$ to the other parties \cite{BennettTeleportation,TeleportationReview}.
Because of Eqs.~(\ref{eq:mu}), the parties share exactly the right
amount of entanglement for this procedure. The overall process achieves
the transformation $\ket{\mathrm{GHZ}}\bra{\mathrm{GHZ}}\rightarrow\sigma$
at rate $r$.
Finally, by inserting Eqs.~(\ref{eq:mu}) in Eq.~(\ref{eq:EMulti-partite}),
we see that the parameter $r$ can take any value compatible with
the inequality 
\begin{equation}
r\leq\frac{1}{E_{\mathrm{c}}^{A|BC\ldots Z}(\sigma)+S(\sigma^{C})+\cdots+S(\sigma^{Z})},
\end{equation}
which completes the proof of Eq.~(\ref{eq:RMulti-partite}). 
\end{document}